\documentclass[a4paper,reqno]{amsart}
\pdfoutput=1

\pdfoutput=1

\usepackage[
  hidelinks,
  bookmarks,
  bookmarksdepth=10,
  bookmarksopen,
  pdfpagemode=UseNone,
  unicode,
  pdftitle={Quotient Inductive-Inductive Types},
  pdfauthor={Thorsten Altenkirch, Paolo Capriotti, Gabe Dijkstra, Nicolai Kraus, and Fredrik Nordvall Forsberg}
  ]{hyperref}
\newcommand{\doi}[1]{doi:\href{https://doi.org/#1}{%
    \urlstyle{same}\nolinkurl{#1}}}
\usepackage{natbib}
\setcitestyle{numbers,square,semicolon,aysep={},yysep={,}}

\usepackage[utf8]{inputenc}
\usepackage[T1]{fontenc}
\usepackage{xypic}
\usepackage{amsmath}
\usepackage{amssymb}
\usepackage{mathtools}
\usepackage{microtype}
\usepackage[capitalise]{cleveref}

\newcommand{\A}{\mathcal{A}}
\newcommand{\B}{\mathcal{B}}
\renewcommand{\C}{\mathcal{C}}
\newcommand{\D}{\mathcal{D}}
\newcommand{\I}{\mathcal{I}}

\newcommand{\Set}{\mathsf{Set}}
\newcommand{\Fam}{\mathsf{Fam}}
\newcommand{\id}{\mathsf{id}}
\newcommand{\Eq}{\mathsf{Eq}}
\newcommand{\Nat}{\mathbb{N}}

\newcommand{\toF}{\Rightarrow}

\newcommand{\node}{\mathsf{node}}
\newcommand{\leaf}{\mathsf{leaf}}

\newcommand{\Con}{\mathsf{Con}}
\newcommand{\Ty}{\mathsf{Ty}}
\newcommand{\Tms}{\mathsf{Tms}}
\newcommand{\Tm}{\mathsf{Tm}}
\newcommand{\emptyCon}{\varepsilon}
\newcommand{\extCon}{\mathsf{ext}}
\newcommand{\baseTy}{\iota}
\newcommand{\sigmaTy}{\sigma}
\newcommand{\sigmaEqTy}{\sigma_{\mathrm{eq}}}

\newcommand{\Alg}[2]{\ensuremath{#1.#2}}
\newcommand{\elt}[2]{\ensuremath{\int^{#1} \hspace{-0.3em} {#2}}}

\newcommand{\barAlg}[1]{#1}

\newcommand{\pathOver}[1]{\;\ensuremath{=\!\!\![#1]}\;}
\newcommand{\ap}[2]{\ensuremath{\mathsf{ap}\;#1\;#2}}

\newcommand{\ct}{%
  \mathchoice{\mathbin{\raisebox{0.5ex}{$\displaystyle\centerdot$}}}%
             {\mathbin{\raisebox{0.5ex}{$\centerdot$}}}%
             {\mathbin{\raisebox{0.25ex}{$\scriptstyle\,\centerdot\,$}}}%
             {\mathbin{\raisebox{0.1ex}{$\scriptscriptstyle\,\centerdot\,$}}}
}

\newcommand{\unit}{\mathbf{1}}

\newcommand{\bool}{\mathbf{2}}

\newcommand{\One}{\mathbf{1}} %terminal cat

\newcommand{\defeq}{\vcentcolon\equiv}
\newcommand{\R}{\mathbb R}

\newcommand{\proofDone}{}

\newtheorem{theorem}{Theorem}[]
\newtheorem{lemma}[theorem]{Lemma}
\newtheorem{corollary}[theorem]{Corollary}
\theoremstyle{definition}
\newtheorem{definition}[theorem]{Definition}
\newtheorem{example}[theorem]{Example}
\theoremstyle{remark}
\newtheorem{remark}[theorem]{Remark}

\begin{document}

\title{Quotient Inductive-Inductive Types}

\author[Altenkirch, Capriotti, Dijkstra, Kraus, Nordvall Forsberg]{%
Thorsten Altenkirch
\and
Paolo Capriotti
\and
Gabe Dijkstra
\and
Nicolai Kraus
\and
Fredrik Nordvall Forsberg%
}

\begin{abstract}
\pdfoutput=1
Higher inductive types (HITs) in Homotopy Type Theory (HoTT) allow the
definition of datatypes which have constructors for equalities over
the defined type. HITs generalise quotient types, and allow to define
types which are not sets in the sense of HoTT (i.e.\ do not satisfy
uniqueness of equality proofs) such as spheres, suspensions and the
torus. However, there are also interesting uses of HITs to define
sets, such as the Cauchy reals, the partiality monad, and the
well-typed syntax of type theory.  In each of these examples we
define several types that depend on each other mutually, i.e.\ they are
inductive-inductive definitions. We call those HITs quotient
inductive-inductive types (QIITs).
Although there has been recent progress on a general theory of HITs,
there is not yet a theoretical foundation for the combination of
equality constructors and induction-induction, despite having many
interesting applications. In the present paper we present a first step
towards a semantic definition of QIITs. In particular, we give an
initial-algebra semantics and show that this is equivalent to the
section induction principle, which justifies the intuitively expected
elimination rules.
%\keywords{Homotopy Type Theory,
%Higher Inductive Types,
%Categorical semantics,
%Induction-Induction.}
\end{abstract}

\maketitle

\section{Introduction}\label{sec:intro}
\pdfoutput=1

This paper is about Type Theory in the sense of Martin-L\"of
\cite{martinlof1972ITT}, a theory which proof assistants such as
Coq~\cite{coq} and Lean~\cite{moura:lean} as well as programming languages such as Agda~\cite{agda}
and Idris~\cite{idris} are based on. Recently, Homotopy Type Theory
(HoTT)~\cite{Univalent2013} has been introduced inspired by homotopy
theoretic interpretations of Type Theory by Awodey and
Warren~\cite{awodeyWarren2009id} and
Voevodsky~\cite{kapulkinLumsdaine2016SSet,voevodsky2010Uni}.

One of the central concepts in Type Theory are inductive definitions,
which allow us to define inductive datatypes like the natural numbers,
lists and infinite trees just by presenting constructors which use the
inductive type in a strictly positive way. Using the propositions as
types explanation we can use the same mechanism to inductively define
predicates and relations like less than or equal, or the derivability
predicate for a logic defined by rules. Conceptually, HoTT changes what
we mean by an inductive definition, because we view a type not only as
given by its elements (points) but also by its equality types
(paths). Hence an inductive definition may not only feature
constructors for elements but also for equality. This concept of a
higher inductive type (HIT) has been used to represent the homotopical
structure of geometric objects, like circles, sphere and tori, and
gives rise to synthetic homotopy theory in HoTT~\cite{shulmanSynth}.

However, as already noted in the HoTT Book~\cite{Univalent2013},
higher inductive types have also more mundane applications, such as
the definition of the Cauchy reals in a way that avoids the use of the
axiom of choice when proving properties like Cauchy completeness of
the reals. Here instead of defining the real numbers as a quotient of
sequences of rational numbers, a HIT is used to define them as the
Cauchy completion of the rational numbers. Similarly, a definition of
the partiality monad which represents potentially diverging operations
over a given type was given using a HIT~\cite{alt-dan-kra:partiality},
again avoiding the axiom of choice when showing for example that the
construction is a monad~\cite{Chapman2015}.

An important observation is that the idea of generating points and
equalities of a type inductively is interesting, even if we do not
care about the higher equality structure of types, or if we do not
want such non-trivial higher structure, for example if we stay inside
the universe $\Set$.  To see this, let us look at an example: consider
trees branching over an arbitrary type, quotiented by arbitrary
permutations of subtrees. We first define the type $T_0(X)$ of
$X$-branching trees, given by the constructors
\begin{equation*}
\begin{alignedat}{3}
 & \leaf_0 : &\;& T_0(X) \\
 & \node_0 : && (X \to T_0(X)) \to T_0(X).
\end{alignedat}
\end{equation*}
We then form the relation $R : T_0(X) \times T_0(X) \to \Set$ that we
want to quotient by as follows: $R$ is the smallest relation such that
for any auto-equivalence on $X$ (i.e.\ any $e : X \to X$ which has
an inverse) and $f : X \to T_0(X)$, we have a proof
$p_{f,e} : R(\node_0(f), \node_0(f \circ e))$, and, secondly, for
$g,h : X \to T_0(X)$ such that $(n : X) \to R(g(n), h(n))$, we have a
proof $c_{f,g} : R(\node_0(g),\node_0(h))$. We can then form the
quotient type $T_0(X) / R$, which is the type of unlabelled trees
where each node has an $X$-indexed family of subtrees, and two trees
which agree modulo the ``order'' of its subtrees are equal.  For
$X \equiv \bool$, these are binary trees where the order of the two
subtrees of each node does not matter.

Now, morally, from a family $X \to (T_0(X) / R)$, we should be able to
construct an element of the quotient $T_0(X) / R$.  This is indeed
possible if $X$ is $\bool$ or another finite type, by applying the
induction principle of the quotient type $X$ times.  However, it seems
that, for a general type $X$, this would require the axiom of
choice~\cite{Univalent2013}, which unfortunately is not a constructive
principle~\cite{diaconescu:ac}.
But using a higher inductive type, we can give an alternative definition for the type of $A$-branching trees modulo permutation of subtrees.
\begin{example} \label{ex:trees}
 Given a type $A$, we define $T(A) : \Set$ by
\begin{equation*}
\begin{alignedat}{3}
  & \leaf : &\;& T(A) \\
  & \node : && (A \to T(A)) \to T(A) \\
  & \mathsf{mix} : && (e : A \to A) \to \mathsf{isEquivalence}(e) \to (f : A \to T(A)) \\ 
  &                && \phantom{(e : A \to A) \to \mathsf{isEquivalence}(e)}  \to \node(f) = \node(f \circ e).
\end{alignedat}
\end{equation*}
\end{example}
Note that in the above example, a set-truncation constructor is
implicitly included in the statement $T(A) : \Set$, which ensures that
$T(A)$ lives in $\Set$. The construction we were looking for is now
directly given by the constructor $\node$.  This demonstration of the
usefulness of higher inductive constructions to increase the strength
of quotients was first discussed in~Altenkirch and
Kaposi~\cite{Altenkirch2016}, where such set-truncated HITs are called
\emph{quotient inductive types} (QITs).

Another example of the use of  higher inductive types is \emph{type theory in type theory}~\cite{Altenkirch2016}, where the well-typed syntax of type theory is implemented as a higher inductive-inductive type in type theory itself.
A significantly simplified version of this will serve as a running example for us:
\begin{example}
\label{ex:tt-in-tt}
We define the syntax of a (very basic) type theory by constructing
types representing contexts and types as follows.
A set $\Con : \Set$ and a type family $\Ty : \Con \to \Set$ are simultaneously defined by giving the constructors
\begin{equation*}
\begin{alignedat}{3}
 & \emptyCon : &\;& \Con \\
 & \extCon : && (\Gamma : \Con) \to \Ty(\Gamma) \to \Con \\
 & \baseTy : && (\Gamma : \Con) \to \Ty (\Gamma) \\
 & \sigmaTy : && (\Gamma : \Con) \to (A : \Ty(\Gamma)) \to \Ty(\extCon \, \Gamma \, A) \to \Ty(\Gamma) \\
 & \sigmaEqTy : && (\Gamma : \Con) \to (A : \Ty(\Gamma)) \to (B : \Ty(\extCon \, \Gamma \, A)) \\
 &              && \phantom{(\Gamma : \Con) \to (A : \Ty(\Gamma))} \to \extCon \, (\extCon \, \Gamma \, A) \, B =_{\Con} \extCon \, \Gamma \, (\sigmaTy \, \Gamma \, A \, B).
\end{alignedat}
\end{equation*}
For simplicity, we do not consider terms. Contexts are either empty
$\emptyCon$, or an extended context $\extCon\,\Gamma\,A$ representing
the context $\Gamma$ extended by a fresh variable of type $A$. Types
are either the base type $\baseTy$ (well-typed in any context), or
$\Sigma$-types represented by $\sigmaTy\,\Gamma\,A\,B$ (well-typed in
context $\Gamma$ if $A$ is well-typed in context $\Gamma$, and $B$ is
well-typed in the extended context $\extCon\,\Gamma\,A$).  Type theory
in type theory as in~\cite{Altenkirch2016} has plenty of equality
constructors which play a role as soon as terms are introduced.  To
keep the example simple we instead use another equality, stating that
extending a context by $A$ followed by $B$ is equal to extending it by
$\Sigma\, A\, B$.  This equality is given by $\sigmaEqTy$.
Note that it is not possible to list the constructors of $\Con$ and
$\Ty$ separately: due to the mutual dependency, the $\Ty$-constructor
$\sigmaTy$ has to be given in between of the two $\Con$-constructors
$\extCon$ and $\sigmaEqTy$.
\end{example}

Despite a lot of work in the literature making use of concrete
HITs~\cite{pi1S1,EMspaces,gitHIT,cavalloThesis,BlakersMassey,brunerieThesis,RPspaces},
and despite the fact that it is usually on some intuitive level clear
for the expert how the elimination principle for such a HIT can be
derived, giving a general specification and a theoretical foundation
for HITs has turned out to be a major difficulty.  Several approaches
have been proposed, and they do indeed give a satisfactory
specification of HITs in the sense that they cover all HITs which have
been used so far (we will discuss related work in a moment).  However, to the
best of our knowledge, there is no approach which
covers \emph{higher inductive-inductive} definitions such as
Example~\ref{ex:tt-in-tt}.  In a nutshell, the purpose of the current
paper is to remedy this.  We restrict ourselves to sets, i.e.\ to
\emph{quotient inductive-inductive types} (QIITs).  This of course is
a serious restriction since it means that we cannot capture many
ordinary HITs such as the circle $\mathbb{S}^1$.  At the same time,
all higher \emph{inductive-inductive} types that we know of are indeed
sets (the Cauchy reals, the surreal numbers, the partiality monad,
type theory in type theory, permutable trees), and will be instances
of our general specification.  Our framework allows arbitrarily
complicated dependency structures.  In particular, we allow
intermixing of constructors as in Example~\ref{ex:tt-in-tt}.

\subsubsection*{Contributions}

We give a formal specification of \emph{quotient inductive-inductive
  types} with arbitrary dependency structure. This can be
viewed as the generalisation of the usual semantics of inductive types
as initial algebras of a functor to the case of quotient
inductive-inductive types.
We establish conditions on the functorial specification of QIITs that
allow us to conclude that the categories of algebras are
complete. This is important because it allows us to prove the
equivalence of initiality and section-induction, justifying the
expected elimination principles.

% Their elimination principle is characterised both categorically as the
% initial object in a category of algebras and in a way which is closer to the ``intuitive'' induction principle one would expect, and we show that these are equivalent.
% This is performed internal to a type theory itself.
% As a consequence, we cannot talk about judgmental computation rules of
% the QIITs that we specify, but all such rules that one would expect do
% hold up to equality.

\subsubsection*{Related Work}

Sojakova~\cite{sojakova2015hits} shows the correspondence between
initiality and induction (a variant of our \cref{thm:main}) for a
restricted class of HITs called $\mathsf{W}$-suspensions.  Basold,
Geuvers and van der Weide~\cite{weide} introduce a syntactic schema
for HITs without higher path constructors, and derive the elimination
rules for them.  Dybjer and Moeneclaey~\cite{dybjerfinitary} give a
syntactic schema for finitary HITs with at most paths between paths,
and give an interpretation in Hofmann and Streicher's groupoid
model~\cite{groupoidModel}. 
Finally, the work by Lumsdaine and Shulman on the semantics of higher inductive types in model categories~\cite{lumsdaine2017semantics} is somewhat similar to an external version of the approach we take in this paper.

\subsubsection*{Preliminaries}

We work in a standard Martin-L\"of style type theory and assume function extensionality.
We write $\Set$ for a type universe which contains types satisfying UIP (\emph{sets} in the terminology of HoTT), and we mostly work with types of this universe.
Univalence is not needed in our development.
When we talk of a \emph{category}, we mean a precategory in the sense of~\cite{Univalent2013} (all our categories become univalent categories if we assume the univalence axiom).
We write $\C \toF \D$ for functors and $X \to Y$ for functions between types.
Note that $\Set$ refers to both the universe and the obvious category of sets and functions, and consequently, $F : A \to \Set$ is a type family, while $F : \C \toF \Set$ is a functor.
Further, we write $\elt{\C}{F}$ for the \emph{category of elements} of $F$.
Recall that this is the category which as objects has pairs $(X,x)$ of an object $X$ in $\C$ and an element $x : FX$.
For a function $f : X \to Y$ and $z,w : X$, we write $\ap f {} : z = w \to f(z) = f(w)$ for the usual ``application of a function to paths''~\cite[Lemma~2.2.1]{Univalent2013}, ${}^{-1} : x = y \to y = x$ for ``path reversal'', and ${}\ct{} : x = y \to y = z \to x = z $ for  and ``path concatenation''~\cite[Lemmas~2.1.1 and 2.1.2]{Univalent2013}.

\section{Sorts}\label{sec:sorts}

Single inductive (and quotient inductive) sets are simply elements of
$\Set$.  Inductive families~\cite{dybjer1994indfam} indexed over some
fixed type $A$ are families $A \to \Set$.  For the inductive-inductive
definitions we are considering, the situation is more complicated,
since we allow very general dependency structures.  Our only
requirement is that there is no looping dependency, since this is
easily seen to lead to contradictions, e.g.\ we do not allow the
definition of a family $A : B \to \Set$ mutually with a family
$B : A \to \Set$ (whatever this would mean).  Concretely, we will
ensure that the collection of type formation rules (the type
signatures) is given in a valid order, and we refer to the types
involved as the \emph{sorts} of the definition.  Hence our first step
towards a specification of general QIITs is to explain what a valid
specification of the sorts is.

Sorts do not only determine the formation rules of the inductive definitions, but also the types of the eliminators.
To capture this, it is not enough to specify a type of sorts --- in order to take the shape of the elimination rules into account, we need to specify a category.

\begin{definition}[Sorts]
  \label{def:sorts}
  A specification of the \emph{sorts} of a quotient
  inductive-inductive definition of $n$ types is given by a list
 \begin{equation*}
  H_0, H_1, \ldots, H_{n-1},
 \end{equation*}
 where each $H_i$ is a functor $H_i : \C_i \toF \Set$.  Here,
 $\Set$ is the category of sets (in the sense of HoTT, i.e.\ types
 with trivial equality types) and functions, $\C_0 \defeq \One$
 is the terminal category,
 and $\C_{i+1}$ is defined as follows:
 \begin{itemize}
  \item objects are pairs $(X,P)$, where $X$ is an object in $\C_i$, and $P : H_i(X) \to \Set$ is  a family of sets;
  \item a morphism  $(f, g) : (X,P) \to (Y,Q)$ consists of a morphism ${f : X \to Y}$ in $\C_i$, and a dependent function $g : (x : H_i(X)) \to P(x) \to Q(H_i(f)\, x)$ (in $\Set$).
 \end{itemize}
 We say that $\C_n$ is the \emph{base category} for the sort signature $H_0, \ldots, H_{n-1}$.
\end{definition}

The following examples will hopefully make clear the connection
between the specification in \cref{def:sorts} and common classes of
data types.

\begin{example}[Permutable trees] %single type
  \label{ex:sorts-single-type}
  For a single inductive type such as the type of trees $T(A)$ in
  \cref{ex:trees}, the sorts are specified by a single functor
  $H_0 : \C_0 \to \Set$ which maps the single object $\star$ of $\C_0$
  to the unit type $\unit$. Objects in the base category $\C_1$ are
  thus pairs $(\star, W)$, where $W : \unit \to \Set$, and morphisms are
  given by $f : \star \to \star$ in $\One$ (necessarily the
  identity morphism), together with a dependent function
  $g : (\star : \unit) \to A(\star) \to B(\star)$.   It is
  easy to see that this category $\C_1$ is equivalent to the category
  $\Set$.
\end{example}

\begin{example}[The finite types]
 Consider the inductive family $\mathsf{Fin} : \Nat \to \Set$ of finite types.
 Again, this is a single type family, i.e.\ we are in the case $n \equiv 1$.
 We have $H_0(\star) \defeq \Nat$,
 and the base category $\C_1$ is equivalent to the category of $\Nat$-indexed families,
 where objects are families $X : \Nat \to \Set$ and morphisms $\C_1(X,Y)$ are dependent functions $f : (n : \Nat) \to X(n) \to Y(n)$.
\end{example}

\begin{example}[Contexts and types]
\label{ex:sorts-con-ty}
 Let us consider the QIIT $(\Con, \Ty)$ from Example~\ref{ex:tt-in-tt}.
 Here, we need two functors $H_0$, $H_1$, the first corresponding to $\Con$ and the second to $\Ty$.
 The first is given by $H_0(\star) \defeq \unit$ as in Example~\ref{ex:sorts-single-type}, since $\Con$ is a type on its own.
 Next, we need $H_1 : \C_1 \to \Set$.
 Applying the equivalence between $\C_1$ and $\Set$ established in \cref{ex:sorts-single-type}, we define $H_1$ to be the identity functor $H_1(A) \defeq A$, since then $\Ty : H_1(\Con) \to \Set$.
 The base category $\C_2$ is equivalent to the category
 $\Fam(\Set)$, whose objects are pairs $(A, B)$ where
 $A : \Set$ and $B : A \to \Set$, and whose morphisms $(A, B)$ to
 $(A', B')$ consist of functions $f : A \to A'$ together with dependent
 functions $g : (x:A) \to B(x) \to B'(f \, x)$.
\end{example}

\begin{example}[the Cauchy reals]
 Recall that the Cauchy reals in the HoTT book~\cite{Univalent2013} are constructed by simultaneously defining $\R : \Set$ and $\sim : \R \times \R \to \Set$ (we ignore the fact that~\cite{Univalent2013} uses $\mathcal{U}$ instead of $\Set$).
 This time the sorts $H_0, H_1$ are given by $H_0(\star) \defeq \unit$ and $H_1(A) \defeq A \times A$,
 corresponding to the fact that $\sim$ is indexed \emph{twice} over $\R$.
 The base category has (up to equivalence) pairs $(X,Y)$ with $Y : X \times X \to \Set$ as objects, and morphisms are defined accordingly.
\end{example}

\begin{example}[The full syntax of type theory]
  Altenkirch and Kaposi~\cite{Altenkirch2016} give the complete syntax of a basic type theory as a (at that point unspecified) QIIT.
 Although this construction is far too involved to be treated as an example in the rest of this paper (where we prefer to work with the simplified version of \cref{ex:tt-in-tt}), we can give the sort signature $H_0, H_1, H_2, H_3$  of this QIIT.
 Apart from contexts $\Con$ and types $\Ty$, this definition also involves context morphisms $\Tms$ and terms $\Tm$:
\begin{equation*}
\begin{alignedat}{6}
  & \Con : &\; & \Set          & \qquad\qquad & \Tms : &\; & \Con \times \Con \to \Set \\
  & \Ty :  &    & \Con \to \Set &             & \Tm :  &   & \big(\Sigma \Gamma : \Con.\Ty(\Gamma)\big) \to \Set.
\end{alignedat}
\end{equation*}
We have:
  \begin{align*}
   H_0(\star) &\defeq \unit && \mbox{$\C_1 \cong \Set$ as in \cref{ex:sorts-single-type};} \\
   H_1(A) &\defeq A && \mbox{$\C_2 \cong \Fam(\Set)$ as in \cref{ex:sorts-con-ty};} \\
   H_2(A,B) &\defeq A\times A && \mbox{$\C_3$ has objects $(A,B,C)$, where $C : A \times A \to \Set$;} \\
    H_3(A,B,C) &\defeq \Sigma\,A\,B && \mbox{$\C_4$ has objects $(A,B,C,D)$, where $D : \big(\Sigma\,A\,B\big) \to \Set$.}
  \end{align*}
\end{example}

\begin{remark}
  Although we work in type theory also in the meta-theory, we give the
  presentation informally in natural language. Formally, the
  specification of sorts and base categories of \cref{def:sorts} can
  be defined as an inductive-recursive
  definition~\cite{dybjersetzer1999finax} of the list
  $H_0, \ldots, H_n$ simultaneously with a function that turns such a
  list into a category. See Dijkstra~\cite[Section~4.3]{gabeThesis} for details.
 % It is however not necessary to assume that the meta-theory allows induction-recursion, since this can be encoded in a standard way.
 % Details can be found in the PhD thesis of Dijkstra~\cite{gabeThesis}.
\end{remark}

The main result of this section states that every base category of a
is complete, i.e.\ it has all small limits. By a small limit, we mean
a limit of a diagram $D : \I \to \C$ where the shape category $\I$ has
a set of objects and each hom-type is a set. This result will be
needed later to show that categories of QIIT algebras are complete.
Recall that $\Set$ has all small limits by a standard construction.

\begin{theorem}[Base categories are complete] \label{thm:sorts-complete}
  For any sort signature
  $H_0, \ldots, H_{n-1}$, the corresponding base category
  $\C_n$ has all small limits.
\end{theorem}
\begin{proof}
 We show that each $\C_k$ is complete by induction on $k$ for $0 \leq k \leq n$.
 Clearly, $\C_0$ is complete.
 For the step case, assume that $\C_k$ is complete.
 By definition, the category $\C_{k+1}$
 has as objects pairs $(X,P)$, where $X$ is an object of $\C_k$ and $P : H_k(X) \to \Set$.
 By the usual correspondence between families and fibrations,\footnote{The correspondence is made precise in~\cite[Thm~4.8.3]{Univalent2013}, but note that the equivalence of the two categories in consideration does not require univalence.}
 we can replace $P$ by a pair $(Y,h)$ of a type $Y : \Set$ and a function $h : Y \to H_k(X)$.
 This means $\C_{k+1}$ is equivalent to the category $\D$ with objects triples $(X,Y,h)$. Morphisms between $(X,Y,h)$ and $(X',Y',k)$
 are triples $(f,g,e)$, where $f : X \to X'$, $g : Y \to Y'$, and $e : k \circ g = F(f) \circ h$.
 
 Consider a diagram $D : \I \to \D$.
 We can split this into three components $D_X : \I \to \Set$, $D_Y : \I \to \Set$, and $D_h : (i : \I) \to D_Y(i) \to H_k(D_X(i))$.
 Consider the cospan
 \begin{equation*}
  \xymatrix{
   &  \lim_\I D_Y \ar[d]^{\lim_\I D_h}  \\
   H_k(\lim_\I D_X) \ar[r] & \lim_\I (H_k \circ D_X)
   }
 \end{equation*}
 where all limits are taken in $\Set$, and where the vertical map is the canonical one given by the universal property of the limit.
 This cospan is itself a diagram in $\Set$, guaranteeing that the pullback of $\lim_\I D_h$ exists; let's call it $\tilde h : \tilde Y \to H_k(\lim_\I D_X)$.
 It is easy to check that $(\lim_\I D_X, \tilde Y, \tilde h)$ is the limit of $D$. \proofDone
\end{proof}

% [Note: At this point, we could explain the correspondence to diagrams over inverse categories in von-Glehn style.]

\section{Algebras}
\label{sec:algebras}

Once the sorts of an inductive definition have been established, the
next step is to specify the \emph{constructors}.  In this section, we
will give a very general definition of constructor specifications, although we will mainly focus on two specific kinds: \emph{point constructors}, which can be thought of as the operations of an algebraic signature, and \emph{path constructors}, which correspond to the axioms.

Similarly to how sorts are specified one by one in \cref{sec:sorts}, we can construct suitable categories of algebras by starting with a finitely complete category $\C$, such as the one obtained from a sort signature, and iteratively adding one constructor at a time. Therefore, we describe this process by describing how to specify a constructor on a finitely complete category $\C$, and show how to extend $\C$ using this constructor specification to get a new finitely complete category $\C'$. Once all constructors have been added, we obtain the sought-after inductive type as the underlying set of an initial object of the category at the last stage, provided this initial object exists.
In the case of the inductive definition of natural numbers, this process will turn out as follows:
\begin{itemize}
  \item we start with $\Set$ as our base category (only one trivial sort, as in Example~\ref{ex:sorts-single-type});
  \item we add a point constructor for the constant corresponding to 0; the category of algebras at this stage is the category of pointed sets;
  \item we add a second point constructor for the operation corresponding to $\mathsf{suc}$; the objects of the category of algebras at this stage are sets equipped with a point and a unary operation;
  \item the set of natural numbers, together with its usual structure, can now be regarded as an initial object in the category of algebras just constructed.
\end{itemize}

\subsection{Relative Continuity and Constructor Specifications}

Roughly speaking, constructors at each stage are given by pairs of $\Set$-valued functors $F$ and $G$ on $\C$, where $G$ is continuous (i.e.\ preserves all small limits).  The intuition is that $F$ specifies the arguments of the constructor, while $G$ determines its target.  For instance, in the example of the natural numbers when specifying the constructor $\mathsf{suc} : \Nat \to \Nat$, $\C$ is the category of pointed sets, and both $F$ and $G$ are the forgetful functor to $\Set$.  The continuity condition on $G$ is needed to make sure that the corresponding category of algebras is complete.  Intuitively, this expresses the idea that a constructor should only ``construct'' elements of one of the sorts, or equalities thereof.  In particular, a constant functor is usually not a valid choice for $G$.

Unfortunately, this simple description falls short of capturing many of the examples of QIITs mentioned in \cref{sec:intro}.  The problem is that we want $G$ to be able to depend on the elements of $F$.  However, since $F$ is assumed to be an arbitrary functor, its category of elements is not necessarily complete, and so we need to refine the the notion of $G$ being continuous to this case.

\begin{definition}[Relative continuity]
\label{defn:relative-continuity}
Let $\C$ be a category, $\C_0$ a complete category, and $U : \C \toF \C_0$ a functor.  A cone over a small diagram in $\C$ is a \emph{$U$-limit cone}, or \emph{limit cone relative to $U$}, if it is mapped to a limit cone in $\C_0$ by $U$. A functor $\C \toF \Set$ is \emph{continuous relative to $U$} if it maps $U$-limit cones to limit cones in $\Set$.
\end{definition}

In particular, the functor $U$ in \cref{defn:relative-continuity} is continuous relative to itself. Furthermore, if $\C$ is complete and $U$ creates limits, then relative continuity with respect to $U$ reduces to ordinary continuity.
If $\C$ is a complete category, and $F : \C \toF \Set$ is an arbitrary functor, the category $\elt{\C}{F}$ of elements of $F$ is equipped with a forgetful functor into $\C$. We will implicitly consider relative limit cones and relative continuity with respect to this forgetful functor, unless specified otherwise.
Note that if $\C$ is complete and $F$ is continuous, then $\elt{\C}{F}$ is also complete, and relative continuity of functors on $\elt{\C}{F}$ is the same as continuity, as observed above.

We can now give a precise definition of what is needed to specify a
constructor:

\begin{definition}[Constructor specifications]
\label{defn:constructor}
A \emph{constructor specification} on a complete category $\C$ is given by:
\begin{itemize}
  \item a functor $F : \C \toF \Set$, called the \emph{argument} functor;
  \item a relatively continuous functor $G : \elt{\C}{F} \toF \Set$, called the \emph{target} functor.
\end{itemize}
\end{definition}

\begin{example}[Permutable trees]
\label{ex:constructor-trees}
The constructor $\leaf : T(A)$ from \cref{ex:trees} can be specified
by functors $F_0 : \Set \toF \Set$ and
$G_0 : \elt{\Set}{F_0} \toF \Set$, where $F_0(A) \defeq \unit$ and
$G_0(A,x) \defeq A$.  Note how $F_0$ specifies the (trivial) arguments of
$\leaf$, and $G_0$ the target. Next the constructor
$\node : (A \to T(A)) \to T(A)$ can be specified by functors
$F_1 : \Set_{\bullet} \toF \Set$ and
$G_1 : \elt{\Set_{\bullet}}{F_1} \toF \Set$, where $\Set_{\bullet}$ is
the category of pointed sets (we think of the point as the previous constructor
$\leaf$): $F_1$ and $G_1$ are defined as $F_1(X, l) \defeq A \to X$ and
$G_1(X,l,f) \defeq X$, so that
\begin{equation*}
  \node : (f : F_1(T(A),\leaf)) \to G_1(T(A), \leaf, f).
\end{equation*}
\Cref{thm:base-target-continuous} will show that $G_0$ and $G_1$ are
relatively continuous.
\end{example}

\begin{example}[Contexts and types]
\label{ex:constructor-con-ty}
The constructor $\sigmaEqTy$ of type
\begin{equation*}
(\Gamma : \Con)(A : \Ty(\Gamma))(B : \Ty(\extCon \, \Gamma \, A)) 
\to \extCon \, (\extCon \, \Gamma \, A) \, B =_{\Con} \extCon \, \Gamma \, (\sigmaTy \, \Gamma \, A \, B)
\end{equation*}
from \cref{ex:tt-in-tt} is specified in the context of the previous
constructors $\emptyCon$, $\extCon$ and $\sigmaTy$ by functors
$F : \C \toF \Set$ and $G : \elt{\C}{F} \toF \Set$, where $\C$ is the
category of algebras of the previous constructors, with
$F(C,T,\epsilon, ext, s) \defeq \Sigma \Gamma : C . \Sigma A :
T(\Gamma). T(ext \, \Gamma \, A)$, and
\begin{equation*}
  G(C,T,\epsilon, ext, s, \Gamma, A, B) \defeq ext \, (ext \, \Gamma \, A) \, B =_{C} ext \, \Gamma \, (s \, \Gamma \, A \, B).
  \end{equation*}
\Cref{lem:general-equality-relatively-continuous} will show that $G$ is relatively continuous.
\end{example}

Given a constructor specification, we can define a the corresponding
category of algebras. In \cref{thm:algebras-complete}, we will see
that the assumptions of \cref{defn:constructor} guarantee that this
category is complete.

\begin{definition}[Category of algebras]
\label{defn:algebras}
Let $(F, G)$ be a constructor specification on a complete category $\C$. The \emph{category of algebras} of $(F, G)$ is denoted $\Alg{\C}{(F,G)}$, and is defined as follows:
\begin{itemize}
  \item objects are pairs $(X, \theta)$, where $X$ is an object of $\C$, and
  $\theta: (x : F X) \to G(X, x)$
  is a dependent function (in $\Set$);
  \item morphisms $(X, \theta) \to (Y, \psi)$ are given by morphisms $f : X \to Y$ in $\C$, with the property that for all $x : FX$,
  \begin{equation*} \psi(F(f)\,x) = G(\overline{f}) (\theta\,x), \end{equation*}
  where $\overline{f} : (X, x) \to (Y, F(f)\,x)$ is the morphism in $\elt{\C}{F}$ determined by $f$.
%   \item composition and identity inherited from $\C$.
\end{itemize}
\end{definition}

We think of $\Alg{\C}{(F,G)}$ as a category of ``dependent
dialgebras''~\cite{hagino}. Note that there is an obvious forgetful
functor $\Alg{\C}{(F,G)} \to \C$.
%Completeness of $\Alg{\C}{c}$ will be proved later in this section ().

\begin{example}[Permutable trees]
\label{ex:algebras-trees}
The category of algebras for the constructor specification
$(F_1, G_1)$ for $\node$ from \cref{ex:constructor-trees} is
equivalent to the category whose objects are triples $(A, l, n)$ where
$A : \Set$, $l : A$, and $n : (X \to A) \to A$.  After specifying also
the $\mathsf{mix}$-constructor, the new category of algebras contains
as well a dependent function $p : (f : A \to T) \to (\sigma : A \cong A) \to n(f) = n(f \circ \sigma)$.
\end{example}

\begin{example}[Contexts and types]
\label{ex:algebras-con-ty}
Similarly, the category of algebras for the constructor specification
from Example~\ref{ex:constructor-con-ty} has objects tuples $(C, T, e, c, b, s, s_{\mathrm{eq}})$ where $(C, T, e, c, b, s)$ is an algebra for the previous constructors, and
\begin{equation*}
  s_{\mathrm{eq}} : (\Gamma : C) \to (A : T(\Gamma)) \to (B : T(c \, \Gamma \, A)) \to c \, (c \, \Gamma \, A) \, B =_{C} c \, \Gamma \, (s \, \Gamma \, A \, B).
\end{equation*}
\end{example}

\subsection{Point Constructors}

If $\C$ is the base category for a sort signature as in \cref{def:sorts}, we can define specific target functors $\C \toF \Set$ which are guaranteed to be relatively continuous.  Constructors having those as targets are referred to as \emph{point constructors}.
Intuitively, a point constructor is an operation that returns an element (point) of one of the sorts. The corresponding target functor is the forgetful functor that projects out the chosen sort. However, sorts can be dependent, so such a projection needs to be defined on a category of elements.

Specifically, let $\C$ be a finitely complete category, $H : \C \toF \Set$ a functor, and $\C'$ the extended base category with one more sort indexed over $H$.  Recall that the objects of $\C'$ are pairs $(X, P)$, where $X$ is an object of $\C$, and $P$ is a family of sets indexed over $H X$.  Let $V_H : \C' \toF \C$ be the forgetful functor.
We define the \emph{base target} functor corresponding to $H$ to be the functor $U_H : \elt{\C'}{(H \circ V_H)} \toF \Set$ given by
\begin{equation*}
  U_H (X, P, x) = P(x).
\end{equation*}
In other words, given an object $X$ of $\C$, a family $P$ over $H X$, and a point $x$ in the base, the functor $U_H$ returns the fibre of the family $P$ over $x$.  The action of $U_H$ on morphisms is the obvious one.

\begin{example}[Permutable trees]
\label{ex:point-constructor-trees}
In \cref{ex:constructor-trees}, the functor
$G_0 : \elt{\Set}{F_0} \toF \Set$ is the composition of the forgetful
$\elt{\Set}{F_0} \toF \Set$ with the base target functor for the only
sort, which in this case is the identity $\id : \Set \toF \Set$.
\end{example}
%
% In the example of natural numbers, we have only one sort, hence $\C = 1$, $H$ is the functor returning $1$, so $\C'$ is just $\Set$.  The corresponding target functor $U_H$ is simply the identity $\Set \to \Set$.  If we take $F$ to be the constant functor returning $1$, and $G$ to be $U_H$, we get the first constructor of the inductive definition of natural numbers.
%
Note that $U_H = \id$ in \cref{ex:point-constructor-trees} is relatively
continuous, as required by \cref{defn:constructor}.  In the rest of
this section, we will show that this is true in general.
Given a category $\C$ and a functor $F : \C \toF \Set$, it is well known that the slice category over $F$ of the functor category $\C \toF \Set$ is equivalent to the functor category $\elt{\C}{F} \toF \Set$. Given a functor $G : \C \toF \Set$ and a natural transformation $\alpha : G \to F$, we will refer to the functor $\overline G : \elt{\C}{F} \toF \Set$ corresponding to $\alpha$ as the \emph{functor of fibres} of $\alpha$. Concretely, $\overline G$ maps an object $(X, x)$, where $x : F X$, to the fibre of $\alpha_X$ over $x$.

\begin{lemma}[auxiliary, not listed in the main body]\label{lem:fibre-relative-continuity}
Let $\C$ be complete, $F,G : \C \to \Set$ functors, and $\alpha : G \to F$ a natural transformation, with functor of fibres $\overline G$. Then $\overline G$ is relatively continuous if and only if, for all small diagrams $X : \I \to \C$ and all limit cones $L \to X$ in $\C$, the following diagram
\begin{equation}\label{eq:pullback-total-space}
\begin{gathered}
\xymatrix{
  GL \ar[r] \ar[d] &
  \lim GX \ar[d] \\
  FL \ar[r] &
  \lim FX
}
\end{gathered}
\end{equation}
is a pullback.
\end{lemma}
\begin{proof}
A generic relative limit cone on $\C$ is determined by a limit cone $\pi : L \to X$, where $X : \I \to \C$ is any small diagram, and an element $z : F L$. Relative continuity of $\overline G$ is equivalent to the map
\begin{equation*}
  \overline G(L, z) \to \lim_i \overline G(X_i, F\pi_i z)
\end{equation*}
being an isomorphism for all such cones $\pi$ and elements $z$.
By the explicit description of the functor of fibres $\overline G$, and the fact that pullbacks commute with limits, we have the pullback squares
\begin{equation}\label{eq:pullback-lim-fibre}
\begin{gathered}
\xymatrix{
  \overline G(L, z) \ar[r] \ar[d]  &   GL \ar[d]   &&     \lim_i \overline G(X_i, F\pi_i z) \ar[r] \ar[d]  &  \lim G X \ar[d]   \\
  1 \ar[r]^-z                      &   F L         &&     1 \ar[r]                                         & \lim FX.
}
\end{gathered}
\end{equation}
If we assume that \eqref{eq:pullback-total-space} is a pullback, we can paste squares to get a pullback
\begin{equation*}
\xymatrix{
  \overline G(L, z) \ar[r] \ar[d] &
  \lim G X \ar[d] \\
  1 \ar[r] &
  \lim FX.
}
\end{equation*}
By uniqueness of limits, it must be that $\overline G(L, z) \cong \lim_i \overline G(X_i, F\pi_i z)$, and the fact that the isomorphism is given by the canonical map follows from a straightforward diagram chase.

Conversely, if $\overline G$ is relatively continuous, it follows from the right square in \eqref{eq:pullback-lim-fibre} that the following diagram is a pullback
\begin{equation*}
\xymatrix{
  \overline G(L, z) \ar[r] \ar[d] &
  \lim G X \ar[d] \\
  1 \ar[r] &
  \lim FX.
}
\end{equation*} 
By taking a coproduct over $FL$ and using extensivity of $\Set$, we get the pullback square \eqref{eq:pullback-total-space}, as required.
\proofDone
\end{proof}

\begin{theorem}[Base target functors are relatively continuous] \label{thm:base-target-continuous}
  Let $\C$ be a complete category, $H : \C \toF \Set$ any functor, and $\C'$ the extended base category corresponding to $H$. %as in \cref{defn:base-category-extension}.
  Then the base target functor $U_H$ is relatively continuous.
\end{theorem}
\begin{proof}
Let $\widetilde{U}_H: \C' \to \Set$ be the functor $\widetilde{U}_H(X, P) = (\Sigma x : H X)P(x)$. There is an obvious natural transformation $\theta : \widetilde U_H \to H \circ V_H$ given by the first projection.
Clearly, $U_H$ is the functor of fibres of $\theta$, hence by \cref{lem:fibre-relative-continuity} all we need to show is that $\theta$ maps limit cones in $\C'$ to pullback squares, which follows immediately from the construction of limits in $\C'$. \proofDone
\end{proof}

\subsection{Reindexing Target Functors}
\label{sec:reindexing}

In many cases, we can obtain suitable target functors by composing the desired base target functor with the forgetful functor to the appropriate stage of the base category. When building constructors one at a time, it will follow from \cref{thm:algebras-complete} and \cref{thm:sorts-complete} applied to the previous steps that this forgetful functor is continuous, and the relative continuity of the target functor will follow.
In more complicated examples, composing with a forgetful functor is
not quite enough. We often want to ``substitute into'' or reindex a
target functor to target a specific element. For example, in the context of
\cref{ex:tt-in-tt}, consider a
hypothetical modified $\sigmaTy$ constructor of the form
\begin{equation*}
  \sigmaTy': \big(\Sigma \Gamma : \Con. \Sigma A : \Ty(\Gamma).\Ty(\extCon\,\Gamma\,A)\big) \to \Ty(\extCon\,\Gamma\,A).
\end{equation*}
We want the target functor to return the set
$\Ty(\extCon\,\Gamma\,A)$, and not just $\Ty(x)$ for a new argument
$x$, which is the result of the base target functor. We can obtain the desired target functor as a composition
\begin{equation*}
\xymatrix{
\elt{\C}{F} \ar[r]^-S & \elt{\Fam(\Set)}{\mathsf{\pi_1}} \ar[r]^-{U_H} & \Set,
  }
\end{equation*}
where $\C$ is the category with objects tuples $(C,T,\epsilon, ext)$, $F : \C \toF \Set$ is the functor giving the arguments of the constructor $\sigmaTy'$, 
$U_H$ is the base target functor corresponding to the second sort, and $S$ is the functor defined by $S(C,T,\epsilon, ext, \Gamma, A, B) \defeq (C, T, ext\,\Gamma\,A)$.

%For example, let us consider a modification of the $\sigma$ constructor (TODO: add reference to example), where we insert a substitution in the return type:
%\begin{equation*}
%  \sigma: (\Gamma : C) \to (A : T\Gamma) \to T (\Gamma.A) \to T(\Gamma.A).
%\end{equation*}
%Here we start with a certain category of algebras $\C$ (the details of which are not particularly important for now), which is complete, and equipped with a forgetful functor $V$ to the base category $\D$, whose objects are pairs $(C, T)$ of a set $C$ and a family $T$ indexed over it.

%The argument functor $F$ maps an object $X$ of $\C$ to the dependent sum
%\begin{equation*}
%  (\Gamma : C) \times (A : T\Gamma) \times T (\Gamma.A),
%\end{equation*}
%where $VX = (C, T)$.  It is easy to see that this defines a functor on $\C$.

%The target functor $G$ is similarly defined: given an object $(X, \Gamma, A, B)$ of the category of elements of $F$, we map it to the set $T (\Gamma.A)$, where again $VX = (C, T)$. Now we see that the target functor is not directly given by the composition of a base target functor with a forgetful functor.  In fact, in this case, we can obtain $G$ as the composition:
%\begin{equation*}
%  \xymatrix{
%    G : \elt{\C}{F} \ar[r]^S & \elt{\D}{\id} \ar[r]^{U_H} & \Set,
%  }
%\end{equation*}
%where $U_H$ is the base target functor corresponding to the second sort, and $S$ is defined by
%\begin{equation*} S(X, \Gamma, A, B) = (C, T, \Gamma.A), \end{equation*}
%and once more $VX = (C, T)$.

Since the functors $S$ that we want to compose with in order to
``substitute'' or reindex are of a special form, the resulting functor
will still be relatively continuous when we start with a relatively
continuous functor. This is made precise by the following result:

%To deal with such examples of ``substituted'' base targets once and for all, we observe that a functor like $S$ has a very special form, which allows us to conclude that composing with it preserves relative continuity.

\begin{lemma}[Preservation of relative limit cones]
\label{lem:substitution-lemma}
Suppose given a commutative diagram of categories and functors
\begin{equation*}
\xymatrix{
  \A \ar[r]^F \ar[d]_{U'} &
  \B \ar[d]^{V'} \\
  \C \ar[r]_G \ar[d]_U &
  \D \ar[d]^V \\
  \C_0 & \D_0,
}
\end{equation*}
where $\C_0$ and $D_0$ are complete, and $G$ maps $U$-limit cones to $V$-limit cones.  Then $F$ maps $(U \circ U')$-limit cones to $(V \circ V')$-limit cones.  In particular, if $\C$ and $\D$ are complete and $G$ is continuous, then $F$ preserves relative limit cones.
\end{lemma}
\begin{proof}
Immediate from the definition of relative limit cone. \proofDone
\end{proof}

\begin{example}
In our example above, we have the following diagram:
\begin{equation*}
  \xymatrix{
    \elt{\C}{F} \ar[r]^-S \ar[d] &
    \elt{\Fam(\Set)}{\mathsf{\pi_1}} \ar[r]^-{U_H} \ar[d] & \Set \\
    \C \ar[r]^-V & \Fam(\Set)
  }
\end{equation*}
where $V : \C \toF \Fam(\Set)$ is the forgetful functor, and hence
continuous.  It follows from the second statement of
\cref{lem:substitution-lemma} that $S$ preserves relative limit cones,
hence $G = U_H \circ S$ is relatively continuous by
\cref{thm:base-target-continuous}.
\end{example}

\subsection{Path Constructors}

Path constructors are constructors where the target functor $G$ returns an \emph{equality} type. They can e.g.\ be used to  express laws when constructing an initial algebra of an algebraic theory as a QIT. We saw an example of this in
\cref{ex:trees}, where we had a path constructor of the form
\begin{equation*}
  \mathsf{mix}: (f : A \to T) \to (\sigma : A \cong A)
             \to \mathsf{node}(f) = \mathsf{node}(f \circ \sigma).
\end{equation*}
The argument functor for $\mathsf{mix}$ is entirely unproblematic. %: starting from a category of algebras for the previous stage, we define a functor $F$ that returns the set of pairs ($f$, $\sigma$) with the types above.
%The target functor, however, must be defined as a $\Set$-valued functor returning a proposition regarded as a set with at most one element, and the problem is showing that such a functor is relatively continuous. This is the aim of the current section.
However, it is perhaps not so clear that the target functor, which sends $(X, l, n, f, \sigma)$ to the equality type $n(f) =_{X} n(f \circ \sigma)$, is relatively continuous. The aim of the current section is to show this for any functor of this form.
Our plan is to start with the prototypical example of such an equality functor, observe that it is relatively continuous, and then show that any other target functor that can occur in a path constructor can be obtained by substitution such that \cref{lem:substitution-lemma} is applicable.

\begin{definition}\label{defn:equality-functor}
  Let $\Eq : \elt{\Set}{(\id \times \id)} \toF \Set$ be the functor defined on objects by
  $\Eq(X, x, y) \defeq x =_X y$ and on morphisms by
  $\Eq(f, p_x, p_y) \defeq p_x \ct (\ap{f}{-}) \ct p_y^{-1}$.
  %  Let $\id : \Set \to \Set$ the identity functor, and denote by $\id \times \id$ the (pointwise) product of $\id$ with itself. There is an obvious \emph{diagonal} natural transformation $\Delta : \id \to \id \times \id$. The functor of fibres of $\Delta$ is denoted
%\begin{equation*} \Eq : \elt{\Set}{(\id \times \id)} \to \Set \end{equation*}
%and is called the \emph{standard equality functor}.
\end{definition}

It is not hard to see that $\Eq$ is a functor. Furthermore, $\Eq$ is
the functor of fibres of the obvious diagonal natural transformation
$\Delta : \id \to \id \times \id$.

%\begin{definition}\label{defn:equality-functor}
%Let $\id : \Set \to \Set$ the identity functor, and denote by $\id \times \id$ the (pointwise) product of $\id$ with itself. There is an obvious \emph{diagonal} natural transformation $\Delta : \id \to \id \times \id$. The functor of fibres of $\Delta$ is denoted
%\begin{equation*} \Eq : \elt{\Set}{(\id \times \id)} \to \Set \end{equation*}
%and is called the \emph{standard equality functor}.
%\end{definition}
%
%Using the explicit definition of a functor of fibres, we see that $\Eq$ maps a bipointed set $(X, x, x')$ to the type of equalities $x = x'$. In particular, $\Eq(X, x, x)$ is a one-element set, and $\Eq(X, x, x')$ is empty when $x \neq x'$.

\begin{lemma} \label{lem:equality-continuous}
The standard equality functor is relatively continuous.
\end{lemma}
\begin{proof}
Both $\id$ and $\id \times \id$ are representable, and in particular continuous. Therefore, the horizontal maps in the diagram of \cref{lem:fibre-relative-continuity} are isomorphisms, which implies that the square is a pullback, hence $\Eq$ is relatively continuous. \proofDone
\end{proof}

With the help of this lemma, one can prove that a large class of equality functors are suitable targets for constructors:

\begin{theorem}[Equality functors are relatively continuous] \label{lem:general-equality-relatively-continuous}
Let $\C$ be a complete category, $F : \C \toF \Set$ any functor, and $G : \elt{\C}{F} \toF \Set$ a relatively continuous functor. Suppose given two global elements $l, r$ of $G$, i.e.\ natural transformations $1 \to G$. The map
\begin{equation*} \Eq_G(l, r) : \elt{\C}{F} \to \Set, \end{equation*}
defined by $\Eq_G(l, r)(X, x) = (l_{(X,x)} =_{G(X, x)} r_{(X,x)})$,
extends to a relatively continuous functor.
\end{theorem}
\begin{proof}
Define $S : \elt{\C}{F} \toF \elt{\Set}{(\id \times \id)}$ by
$S(X, x) = (G(X, x), l_{(X, x)}, r_{(X, x)})$;
this is a functor since $l$ and $r$ are natural.  Observe that $\Eq_G(l, r)$ can now be obtained as the composition
\begin{equation*}
  \xymatrix{
    \elt{\C}{F} \ar[r]^-S & \elt{\Set}{(\id \times \id)} \ar[r]^-{\Eq} & \Set,
  }
\end{equation*}
hence the conclusion of the lemma will follow once we establish that $S$ preserves relative limit cones. Consider the diagram:
\begin{equation*}
  \xymatrix{
    \elt{\C}{F} \ar[r]^-S \ar[d]_-{\id} &
    \elt{\Set}{\id \times \id} \ar[d] \\
    \elt{\C}{F} \ar[r]^-G \ar[d] &
    \Set \ar[d]^{\id} \\
    \C &
    \Set,
  }
\end{equation*}
which commutes by definition of $S$. Since $G$ is relatively continuous by assumption, it preserves relative limit cones, hence so does $S$ by \cref{lem:substitution-lemma}, as required. \proofDone
\end{proof}

\begin{example}[Permutable trees]
  The target of the $\mathsf{mix}$ constructor from \cref{ex:trees} can
  be obtained as an equality functor in this sense. We take $G$ to be
  the underlying sort, which is relatively continuous by the results
  of the previous section. The global elements $l$ and $r$ are defined
  %at component $(A, l, n, f, \sigma)$
  by $l_{(X, l, n, f, \sigma)} \defeq n(f)$ and
  $r_{(X, l, n, f, \sigma)} \defeq n(f \circ \sigma)$. Their
  naturality can easily be verified directly.
  \end{example}

%We will call functors of the form $\Eq_G(l, r)$ \emph{equality functors}. Our example of path constructor above can now be obtained using an equality functor as target. We take $G$ to be the underlying sort, which is relatively continuous by the results of the previous section (although this is immediate in this particular case, since it is a composition of forgetful functors). The global elements $l$ and $r$ can be defined by just following their expressions as terms: namely, for a given object $X = (T, \mathsf{leaf}, \mathsf{node}, f, \sigma)$, $l_X$ is simply $\mathsf{node}(f)$ and $r_X$ is $\mathsf{node}(f \circ \sigma)$. Their naturality can easily be verified directly. Finally, the whole constructor is given by the pair $(F, \Eq_G(l, r))$.

Iterating equality functors, one can also express \emph{higher} path constructors, but in our limited setting of inductively defined \emph{sets}, there is little reason to go beyond one level of path constructors --- higher ones will have no effect on the resulting inductive type. However, we believe that the ease with which \cref{lem:general-equality-relatively-continuous} can be applied iteratively will be an important feature when generalising our technique to general higher inductive types. We discuss this further in \cref{sec:conclusion}.

\subsection{Categories of Algebras are Complete}

If $\C$ is a complete category, and $(F, G)$ is a constructor specification on $\C$, recall that the category of algebras $\Alg{\C}{(F,G)}$ from \cref{defn:algebras} has ``dependent $(F,G)$-dialgebras'' as objects, and maps that commute with the dialgebra structures as morphisms.
In this section, we will show that $\Alg{\C}{(F,G)}$ is complete, and that its forgetful functor is continuous. The significance of this result is twofold:

First of all, it makes it possible to use the power of limits when reasoning about properties of algebras; in particular, we will show in \cref{sec:elimination} how, using products and equalisers, one can extend the classical equivalence between initiality and induction for ordinary inductive types to our setting.

Secondly, it goes a long way towards establishing an existence result for initial algebras; since a category of algebras over $n + 1$ constructors is complete, and the forgetful functor to the category of algebras over the first $n$ preserves limits, it follows from the adjoint functor theorem that this functor has a left adjoint if and only if it satisfies the solution set condition. Since this can be applied to every stage, we get a left adjoint for the forgetful functor down to $\Set$, and in particular an initial object.
Given our assumptions on constructors, there is no reason to expect the solution set condition to hold at this generality. However, we expect it to follow from an appropriate ``accessibility'' condition on the argument functors. This is discussed further in \cref{sec:conclusion}.

\begin{theorem}[Categories of algebras are complete] \label{thm:algebras-complete}
Let $\C$ be a complete category, and $(F, G)$ a constructor specification on it. Then $\Alg{\C}{(F,G)}$ is complete.
\end{theorem}
\begin{proof}
Let $\widetilde G : \C \to \Set$ be defined by $\widetilde G(X) = (x : F X) \times G(X, x)$, and let $p : \widetilde G \to F$ be the first projection. Then clearly $G$ is the functor of fibres of $p$.
Now consider a diagram $Y : \I \to \Alg{\C}{(F,G)}$. The diagram $Y$ can be decomposed into a diagram $X : \I \to \C$, and a natural transformation $s : FX \to \widetilde GX$ which is a section of $pX$. If $\pi : L \to X$ is a limit cone for $X$, by \cref{lem:fibre-relative-continuity} we get a pullback square
\begin{equation*}
\xymatrix{
  \widetilde GL \ar[r] \ar[d] &
  \lim \widetilde GX \ar[d] \\
  FL \ar[r] &
  \lim FX,
}
\end{equation*}
and $s$ determines a section $\lim s : \lim F X \to \lim \widetilde G X$ of the right vertical morphism in the diagram. Let $t : FL \to \widetilde GL$ be the section of $p_L$ obtained by pulling back $\lim s$. In particular, $t$ has the form $\langle \id, \theta \rangle$, where
\begin{equation*}
  \theta : (x: F X) \to G(X, x),
\end{equation*}
and $\pi$ extends to a cone $(L, \theta) \to Y$ in $\Alg{\C}{(F,G)}$. It is now easy to verify that this is a limit cone. \proofDone
\end{proof}

\section{Elimination Principles}\label{sec:elimination}

So far, we have given rules for specifying a QIIT by
giving a sort signature and a list of constructors.
As type-theoretical rules, these correspond to the formation and
introduction rules for the QIIT. In this section, we introduce he
corresponding elimination rules, stating that a QIIT is the smallest
type closed under its constructors. We show that a categorical
formulation of the elimination rules is equivalent to the universal
property of initiality.

\subsection{The Section Induction Principle}

%A. Joyal. 2014. Categorical homotopy type theory. Slides of a seminar given at MIT. Available from http://ncatlab.org/homotopytypetheory/files/Joyal.pdf. (2014).

The elimination principle for an algebra $\barAlg{X}$ states that
\emph{every fibred algebra over $\barAlg{X}$ has a section}, where a
fibred algebra over $\barAlg{X}$ is an algebra family
``$Q : \barAlg{X} \to \Set$'', and a section of it a dependent algebra
morphism ``$(x : \barAlg{X}) \to Q(x)$''.
The usual correspondence between type families and fibrations
extends to algebras (see the examples below), and so we formulate
the elimination rule for $\barAlg{X}$ as $\barAlg{X}$ being section
inductive in the category of algebras in the following sense:

\begin{definition}[Section inductive]
  An object $X$ of a category $\C$ is \emph{section inductive} if for
  every object $Y$ of $\C$ and morphism $p : Y \to X$, there exists
  $s : X \to Y$ such that $p \circ s = \id_X$.
\end{definition}
For an algebra $\barAlg{X}$, the existence of the underlying
function(s) $\barAlg{X} \to \barAlg{Y}$ corresponds to the elimination
rules, while the fact that they are algebra morphisms corresponds to
the computation rules. %Let us briefly explore what it means to be
%section inductive for our running examples.

\begin{example}[Permutable trees]
  Consider permutable-tree algebras, e.g.\ tuples $(X, l, n, p)$ as in
  \cref{ex:algebras-trees}. A fibred permutable-tree algebra over
  $(X, l, n, p)$ consists of $Q : X \to \Set$ together with $m_l : Q(l)$ and
  \begin{equation*}
  \begin{alignedat}{3}
    & m_n : &\;& (f : A \to X) \to (g : (a : A) \to Q(f\,a)) \to Q(n\,f) \\
    & m_p : &  & (f : A \to X) \to (g : (a : A) \to Q(f\,a)) \to (\sigma : A \cong A) \\ 
    &       &  & \phantom{(f : A \to X) \to (g : (a : A)\to )} \to m_n\,f\,g \pathOver{\ap{Q}{p}} m_n\,(f \circ \sigma)\,(g \circ \sigma)
  \end{alignedat}
  \end{equation*}
  Here the type $x \pathOver{p} y$ is the types of equalities between elements $x : A$ and $y : B$ in different types, themselves related by an equality proof $p : A = B$.
  This data can be arranged into an ordinary algebra
  $\Sigma x : X.Q(x)$, together
  with an algebra morphism $\pi_1 : (\Sigma x : X)Q(x) \to X$. A
  section of this morphism is exactly a dependent function
  $h : (x : X) \to Q(x)$. Since $h$ comes from an algebra morphism, we
  further know that e.g.\ $h(l) = m_l$ and
  $h(n\,f) = m_n\, f\,(h \circ f)$. Conversely, if we start with an
  algebra morphism $g : (X', l', n', p') \to (X, l, n, p)$, this gives
  rise to a fibred algebra $(Q, m_l, m_n, m_p)$ by
  considering the fibres $Q(x) = \Sigma y : A'.g(y) = x$ of
  $p$. The points $m_l$, $m_n$ and the path $m_p$ arise from the proof
  that $g$ preserves $a'$, $b'$ and $p'$.
  % Consider interval algebras, e.g.\ tuples $(A, a, b, p)$ from
  % TODO-ref-missing. A fibred interval algebra over $(A, a, b, p)$
  % consists of $Q : A \to \Set$ together with $m_a : Q(a)$,
  % $m_b : Q(b)$ and $m_p : m_a \pathOver{\ap{Q}{p}} m_b$. This data can
  % be arranged into an ordinary interval algebra
  % $((\Sigma x : A)Q(x), (a, m_a), (b, m_b), \pairEq{p}{m_p})$ together
  % with an algebra morphism $\pi_1 : (\Sigma x : A)Q(x) \to A$. A
  % section of this morphism is exactly a dependent function
  % $f : (x : A) \to Q(x)$. Since $f$ comes from an algebra morphism, we
  % further know that $f(a) = m_a$, $f(b) = m_b$, and $\ap{f}{p} =
  % m_p$. Conversely, if we start with an algebra morphism
  % $p : (A', a', b', p') \to (A, a, b, p)$, this gives rise to a fibred
  % interval algebra $(Q : A \to \Set, m_a, m_b, m_p)$ by considering
  % the fibres $Q(x) = (\Sigma y : A')(p(y) = x)$ of $p$. The points
  % $m_a$, $m_b$ and the path $m_p$ arise from the proof that $p$
  % preserves $a'$, $b'$ and $p'$.
\end{example}

\begin{example}[Contexts and types]
  For context-and-types algebras from \cref{ex:algebras-con-ty}, a fibred algebra
  over $(C, T, e, c, b, s, s_{\mathrm{eq}})$ consists of $Q : C \to \Set$ and $R : (x : C) \to T(x) \to Q(x) \to \Set$, together with $m_e : Q(e)$ and
  \begin{equation*}
  \begin{alignedat}{3}
    & m_c : && (\Gamma : C) \to (x : Q(\Gamma)) \to (A : T(\Gamma)) \to R(\Gamma, A, x) \to Q(c \, \Gamma \, A) \\
    & m_b : && (\Gamma : C) \to (x : Q(\Gamma)) \to R(\Gamma, b\,\Gamma, x) \\
    & m_s : && (\Gamma : C) \to (x : Q(\Gamma)) \to (A : T(\Gamma)) \to (y : R(\Gamma, A, x) \to (B : T(c\,\Gamma\,A))  \\ 
    &       &&  \phantom{(\Gamma : C)} \to (z : R(c\,\Gamma\,A, B, m_c\,\Gamma\,x\,A\,y)) \to R(\Gamma, s\,\Gamma\,A\,B, x) \\
    & m_{s_{\mathrm{eq}}} : && (\Gamma : C) \to (x : Q(\Gamma)) \to (A : T(\Gamma)) \to (y : R(\Gamma, A, x)) \\ 
    &       && \phantom{(\Gamma : C)} \to (B : T(c\,\Gamma\,A)) \to (z : R(c\,\Gamma\,A, B, m_c\,\Gamma\,x\,A\,y)) \\ 
    &       && \phantom{(\Gamma : C)} \to m_c\,(c\,\Gamma\,A)\,(m_c\,\Gamma\,x\,A\,y)\,B\,z \pathOver{\ap{Q}{(s_\mathrm{eq}\,\Gamma\,A\,B)}} \\
    &       && \hspace{14em} m_c\,\Gamma\,x\,(s\,\Gamma\,A\,B)\,(m_s\,\Gamma\,x\,A\,y\,B\,z)
  \end{alignedat}
  \end{equation*}
  Again, this data can be arranged into an ordinary algebra with base
  $C' : \Set$, $T' : C' \to \Set$, where $C' = \Sigma x : C.Q(x)$ and
  $T'(x, q) = \Sigma y : T(x).R(x,y,q)$, together with an algebra
  morphism $(\pi_1, \pi_1) : (C', T') \to (C, T)$. A section of this
  morphism gives functions $f : (x : C) \to Q(x)$ and
  $g : (x : C) \to (y : T(x)) \to R(x, y, f\,x)$ that preserve the
  algebra structure.
\end{example}

A general account of the equivalence between the usual formulation of the
elimination rules and the section induction principle is
in Dijkstra~\cite[Section~5.4]{gabeThesis}.

\subsection{Initiality, and its Relation to the Section Induction Principle}

The section induction principle for an algebra $\barAlg{X}$ matches
our intuitive understanding of the elimination rules for $\barAlg{X}$,
but it is perhaps a priori not so clear that e.g.\ satisfying it
defines an algebra uniquely up to equivalence. In this section, we
show that this is the case by proving that the section induction
principle is equivalent to the categorical property of initiality.

\begin{definition}[Initiality]
  An object $X$ of a category $\C$ is initial if for every object $Y$
  of $\C$, the set of morphisms $X \to Y$ is contractible.
\end{definition}
It is immediate that the property of being initial is a mere
proposition. It is also more or less obvious that initiality implies
section induction:

\begin{lemma} \label{thm:initToSec}
If an object $X$  in a category $\C$ is initial, then it is section inductive.
\end{lemma}
\begin{proof}
  Assume $X$ is initial. Given $p : Y \to X$, we need to
  produce $s : X \to Y$ such that $p \circ s = \id_X$.
  Since $X$ is initial, there is an arrow $s : X \to Y$. Further
  $p \circ s : X \to X$, so by uniqueness of morphisms $X \to X$, we
  have $p \circ s = \id_X$. \proofDone
\end{proof}

For the converse, a little bit more structure in $\C$ is needed:

\begin{lemma} \label{thm:secToInit}
If an object $X$ in a category $\C$ with finite limits is
section inductive, then it is initial.
\end{lemma}
 \begin{proof}
   Given an object $Y$ in $\C$, we need to provide a unique arrow
   $X \to Y$. Consider the projection $\pi_1 : X \times Y \to X$,
   which is an arrow into $X$, and therefore has a section
   $s : X \to X \times Y$. Our candidate arrow is then
   $\pi_2 \circ s : X \to Y$, which we have to show is unique. Using
   equalisers, we can show that any two arrows $f,g$ out of $X$ to
   some other object $Y$ are equal.  Let $E$ be the equaliser of $f$
   and $g$, then we get a projection map $i : E \to X$. By the section
   principle, this map has a section $s : X \to E$:
   \begin{equation*}
     \xymatrix{
       E \ar[r]^{i} &X \ar@<-.5ex>[r]_-{g} \ar@<.5ex>[r]^-{f} &Y \\
       X \ar[u]^{s} \ar[ur]_{\id_{X}} }
   \end{equation*}
   Hence
   $f = \id_X \circ f = s \circ i \circ f = s \circ i \circ g = \id_X \circ g = g$
   holds. \proofDone
 \end{proof}

Using all these ingredients, we get the main theorem of this section:

\begin{theorem}[Initiality $\cong$ section induction] \label{thm:main}
  An object $X$ in a in a category of algebras $\Alg{\C}{(F,G)}$ being
  initial is equivalent to it being section inductive.
\end{theorem}
\begin{proof}
  By \cref{thm:algebras-complete}, the category $\Alg{\C}{(F,G)}$ is complete. Hence by
  \cref{thm:initToSec,thm:secToInit}, the two statements are logically
  equivalent. Since a logical equivalence between two mere
  propositions is automatically an equivalence, and initiality is
  easily seen to be a mere proposition, all that remains is to show
  that the section induction property is a mere proposition. For this,
  we may assume that the type in question is inhabited, and it
  suffices to show that the set of sections
  $\big(\Sigma s : X \to Y\big)\big(p \circ s = \id_X\big)$ is a mere
  proposition for any object $Y$. But since $X$ is initial by
  assumption and \cref{thm:secToInit}, the sets $X \to Y$ and
  $X \to X$ are contractible, hence so is the path type
  $p \circ s =_{X \to X} \id_X$, and we are done. \proofDone
\end{proof}

As an application, we can now reason about QIITs using their
categories of algebras. For instance, we get a short proof of the
following fact:

\begin{corollary}
  The interval is equivalent to the unit type.
\end{corollary}
\begin{proof}
  By \cref{thm:main}, the interval is the initial object in the
  category with objects
  $\Sigma X : \Set.\Sigma x : X.\Sigma y : X.x =_{X} y$, while the unit type is the initial object in the category with objects  $\Sigma X : \Set.X$.
  By singleton contractibility, the former is equivalent to the
  latter, and since initiality is a universal property,
  the two initial objects coincide up to equivalence. \proofDone
\end{proof}

\section{Conclusions and Further Work}\label{sec:conclusion}

We have developed a semantic framework for QIITs: QIITS give rise to a
category of algebras and the initial object of this category represent
the types and constructors of the QIIT. This generalises the usual functorial semantics of
inductive types to this much more general setting. So far we have
verified the appropriateness of this setting by means of examples.
In future work, we would like to explicitly relate the syntax of QIITs to the
corresponding semantics.

Our category of algebras is complete. This is helpful when developing a metatheory of QIITs, as demonstrated by the  proof of equivalence of initiality and section induction (\cref{thm:main}), justifying elimination principles. Of course, completeness %of the category of algebras
is not by itself sufficient to derive the existence of initial algebras, but it suggests that it should be possible to restrict the argument functors enough to guarantee this, possibly by reducing the existence of QIITs to some basic type former playing an analogous role to that of W-types for ordinary inductive types.
We believe that completeness of the categories of algebras will allow an existence proof using the adjoint functor theorem.

We have restricted our attention to QIITs, but we believe that our construction is applicable to general HITs (and even HIITs). While at first glance such an extension of our framework seems to require an internal theory of $(\infty,1)$-categories, we believe that it is enough to keep track of only a very limited number of coherence conditions, making this extension possible even without solving the well-known problem of specifying an infinite tower of coherences in HoTT.

There are other directions of future work one may consider, e.g.\ the combination
of QIITS and induction-recursion, and the possibility of generalising
coinductive types along similar lines. In any case these generalisations should be driven by examples,
similar to how the examples discussed in the current paper have motivated the need for QIITs.

\subsubsection*{Acknowledgements.}
We thank Ambrus Kaposi and Jakob von Raumer for many interesting discussions.
%This research was supported by EPSRC grants EP/M016994/1 (TA, NK) and
%EP/K023837/1 (FNF), as well as AFOSR award FA9550-16-1-0029 (TA, PC).
This research was supported by EPSRC grants EP/M016994/1 and
EP/K023837/1, as well as AFOSR award FA9550-16-1-0029.

\bibliographystyle{plainnat}
\bibliography{qiits}

\end{document}